\documentclass[conference,usletter]{IEEEtran}
\usepackage{stfloats}
\usepackage{color}
\usepackage{amssymb} 
\usepackage{mathtools}
\usepackage{footnote}
\usepackage{amsfonts}
\newcommand{\ve}[1]{{\bf #1}}

\newcommand{\norm}[1]{\left|\left|#1\right|\right|}

\newcommand{\prox}{\text{prox}}

\usepackage{algorithm,algpseudocode}
\usepackage{balance}

\newtheorem{theorem}{Theorem}

\newtheorem{remark}{Remark}

\newtheorem{property}{Property}

\ifCLASSOPTIONcompsoc \usepackage[caption=false,font=normalsize,labelfon t=sf,textfont=sf]{subfig} \else \usepackage[caption=false,font=footnotesize]{subfi
g} \fi

\begin{document}

\sloppy
\IEEEoverridecommandlockouts
%% Paper Title
%% You can use linebreaks \\ within to get better formatting as
%% desired. 
\title{A Sequential Approximation Framework for Coded Distributed Optimization}

%Coding for Distributed Sequential Matrix-Vector Multiplication with Applications}%A coding scheme for distributed matrix multiplications with stragglers} 

%% Author names and affiliations:
%%
%% Avoiding spaces at the end of the author lines is not a problem with
%% conference papers because we don't use \thanks or \IEEEmembership.
%%
%% For several authors with only one affiliation:
%%
\author{

Jingge Zhu, Ye Pu, Vipul Gupta, Claire Tomlin, Kannan Ramchandran\\

EECS,   University of California, Berkeley, CA, USA. \\
      Email: \{jingge.zhu, yepu, vipul\_gupta, tomlin, kannanr\}@eecs.berkeley.edu 

\thanks{The work of Jingge Zhu was supported by the Swiss National Science Foundation under Project P2ELP2\_165137. The work of Ye Pu  was supported by the Swiss National Science Foundation under Project P2ELP2\_165155.}

%\thanks{J. Zhu is with the School of Computer and Communication Sciences, Ecole Polytechnique F{\'e}d{\'e}rale de Lausanne (EPFL), Lausanne,
%Switzerland (e-mail: jingge.zhu@epfl.ch).}
%
%
%\thanks{M. Gastpar is with the School of Computer and Communication Sciences, Ecole Polytechnique F{\'e}d{\'e}rale de Lausanne (EPFL), Lausanne, Switzerland and the Department of Electrical Engineering and Computer Sciences, University of California, Berkeley, CA, USA (e-mail: michael.gastpar@epfl.ch).}
}

%%
%% For over three affiliations, or if they all won't fit within the width
%% of the page, use this alternative format:
%%
% \author{
%   \IEEEauthorblockN{
%     Michael Shell\IEEEauthorrefmark{1},
%     Homer Simpson\IEEEauthorrefmark{2},
%     James Kirk\IEEEauthorrefmark{3}, 
%     Montgomery Scott\IEEEauthorrefmark{3} and
%     Eldon Tyrell\IEEEauthorrefmark{4}}
%   \IEEEauthorblockA{
%     \IEEEauthorrefmark{1}School of Electrical and Computer Engineering\\
%     Georgia Institute of Technology, Atlanta, Georgia 30332--0250\\ 
%     Email: see http://www.michaelshell.org/contact.html}
%   \IEEEauthorblockA{
%     \IEEEauthorrefmark{2}Twentieth Century Fox, Springfield, USA\\
%     Email: homer@thesimpsons.com}
%   \IEEEauthorblockA{
%     \IEEEauthorrefmark{3}Starfleet Academy, San Francisco, California 96678-2391\\
%     Telephone: (800) 555--1212, Fax: (888) 555--1212}
%   \IEEEauthorblockA{
%     \IEEEauthorrefmark{4}Tyrell Inc., 123 Replicant Street, Los Angeles, California 90210--4321}
% }

%% Use for special paper notices

%% To balance the two columns, you should reduce the text-height of
%% the last page using the following command:
%%%%%%%%%%%%%%%%%%%%%%%%%%%%%%%%%%%%%%%%%%%%%%%%%%%%%%%%%%%%%%%%%%%%%
%\addtolength{\textheight}{-9.35cm}
%%%%%%%%%%%%%%%%%%%%%%%%%%%%%%%%%%%%%%%%%%%%%%%%%%%%%%%%%%%%%%%%%%%%%
%% with an appropriate value. This command must be place on the second
%% last page, i.e., for a one-page abstract here, for a two-page
%% abstract right after the \maketitle command.

%% Create the title:
\maketitle

%% Abstract: 
%% For the final version of the accepted paper, please make sure you
%% remove the comment "THIS PAPER IS ELIGIBLE FOR THE STUDENT PAPER
%% AWARD."
%%
\begin{abstract}

Building on the previous work of Lee \textit{et al.} \cite{lee_speeding_2015} and Ferdinand \textit{et al.} \cite{ferdinand_anytime_2016} on coded computation, we propose a sequential approximation framework for solving optimization problems in a distributed manner.   In a distributed computation system, latency caused by individual processors (``stragglers") usually causes a significant delay in the overall process.  The proposed method is powered by a sequential computation scheme, which is designed specifically for systems with stragglers. This scheme has the desirable property that the user is guaranteed to receive useful (approximate) computation results whenever a processor finishes its subtask, even in the presence of uncertain latency.  In this paper, we give a coding theorem for sequentially computing matrix-vector multiplications, and the optimality of this coding scheme is also established. As an application of the results, we demonstrate solving optimization problems using a sequential approximation approach, which accelerates the algorithm in a distributed system with stragglers.  %The empirical results show that the gains are most significant when we aim to achieve an approximate solution. %When solving some instances of the lasso problem, the typical saving on execution time ranges from  $30\%-50\%$ with a normalized suboptimality of $0.5$.

\end{abstract}

\section{Introduction}
\label{sec:intro}

Emerging applications from social networks and machine learning make distributed computation systems increasingly important for handling large-scale computation tasks.  In this framework, a large computation task is divided into several smaller sub-tasks, each of which is dispatched to a different processor. The computation results are then aggregated and processed to produce the final result.  A central challenge to this approach is how to handle  uncertainty caused by ``system noise" (see, 	e. g. \cite{dean_tail_2013}). One notable phenomenon is the ``straggler" effect, namely, the latency of a single processor could cause a significant delay in the whole computational task. In existing distributed computation schemes, various straggler-detecting algorithms have been proposed to mitigate this problem. For example, Hadoop detects stragglers when executing the computation. When it detects a straggler, it runs a copy of the ``straggled" task on a different processor. However, running many replications of subtasks turn out to be  inefficient.

A novel approach to mitigating uncertainty is to add controlled redundancy in the distributed computation tasks.  Lee \textit{et al.} \cite{lee_speeding_2015} proposed a coded computation framework for computing matrix-vector multiplications in a distributed system.  By using maximum distance separable (MDS) codes to encode the matrix and distribute smaller computation tasks to different processors, they show that coded computation can derive significant gains over na\"{i}ve replication methods in terms of computation time. Based on the same idea, Ferdinand and Draper  \cite{ferdinand_anytime_2016} proposed  a refined coding scheme (called ``anytime coding scheme") where an approximation of the matrix-vector multiplication can be obtained  in a timely fashion.  We  point out that the coded computation scheme has also been extended to study matrix multiplication problems~\cite{yu_polynomial_2017}, and is shown to be useful for reducing the communication overhead in distributed systems \cite{li_coded_2015}. Furthermore, the idea of using codes in distributed systems has found  various applications in machine learning problems, as shown in \cite{tandon_gradient_2016} \cite{dutta_short-dot} \cite{Karakus_2017}.

In this paper, we take a step further in studying how to tackle  optimization problems using a coded computation approach. Building on the work of \cite{lee_speeding_2015} and \cite{ferdinand_anytime_2016}, we propose a \textit{sequential approximation method} for solving optimization problems. The basic idea of this approach is that  instead of directly solving the original problem, we solve a sequence of optimization problems (called \textit{approximations}), whose solutions gracefully approach the solution of the original problem. These approximations need to be designed judiciously so that  solving the approximate problems requires less computation time than solving the original problem.  Consequently as we show in the sequel, in the presence of stragglers, the sequential approximation method typically takes less time to find the solution  of the original problem. %In our example of solving instances of lasso problems, the saving on execution time typically range from $30\%-50\%$ if we aim to achieve a normalized suboptimality of $0.5$ (see details in Section \ref{sec:application}).
The saving on execution time is more significant if we only aim to find an approximate solution to the original problem. An attractive feature of the proposed method is that the processors in the distributed system are oblivious to different approximations, making it a user-centered design.

The driving mechanism for the proposed sequential approximation framework is a so-called \textit{coded sequential computation scheme}, designed specifically for  distributed computation systems with latency. It has the desirable property that the user is guaranteed to receive useful (approximate) computation results whenever a subset of processors finish their subtasks, even in the presence of uncertain latency. In this paper, we focus our study on a   coding scheme for sequentially computing  matrix-vector multiplication, which is a basic building block for most algorithms. We then show how to integrate our coded sequential  computation scheme into the sequential approximation framework in order to accelerate the algorithm in the distributed computation system.

\section{The Sequential Approximation Method: an Overview}
\label{sec:overview}

We consider solving an optimization problem of the following form
\begin{align}
\text{minimize } (1/2)\ve x^T\ve H\ve x+\ve h^T\ve x+g(\ve x)
\label{eq:objective}
\end{align}
where $\ve H\in\mathbb R^{m\times m}$ is positive (semi)-definite matrix and $g:\mathbb R^m\mapsto \mathbb R$  is a closed, proper and convex function.   The formulation in (\ref{eq:objective}) represents a large class of problems of interests. For example,  choosing $g(\ve x)$ to be the indicator function\footnote{The indicator function $\ve 1_{\ve x\in\mathbb K}$ evaluates to $0$ if $\ve x\in\mathbb K$ and to $\infty$ if $\ve x\notin\mathbb K$.} $g(\ve x) = \ve 1_{\ve x\in\mathbb K}$ converts problem (\ref{eq:objective}) into  a constrained optimization problem with a qudratic objective function, where $\mathbb K\subset\mathbb R^n$ is a  convex set. Choosing $g(\ve x)$ to be a norm of $\ve x$ is also widely used in applications. For example, the choice $g(\ve x) = \gamma\norm{\ve x}_1$ converts (\ref{eq:objective}) to the lasso problem. 

\textit{Alternating methods} \cite{} are efficient optimization methods used to solve  problems of the form (\ref{eq:objective}). The proximal gradient method \cite{parikh_proximal_2014} and ADMM are two examples of the alternating methods. For instance, the proximal gradient method update the variable $\ve x$ as
\begin{align}
\ve x^{k+1}&=\prox_{t^kg}(\ve x^{k}-t^k (\ve H\ve x^{k}+\ve h))\label{eq:proximal}
%&:=f(\ve H\ve x^k)\nonumber
\end{align}
where $\ve x^{k}$ and $t^k$ denote the variable and the step size in the $k$-th iteration, respectively. The proximal operator $\prox$ is defined as
\begin{align*}
\prox_{tg}(\ve v):=\text{argmin}_{\ve x}(g(\ve x)+(1/2t)\norm{\ve x-\ve v}_2^2).
\end{align*}
The update rule of ADMM is given as
\begin{align}
\ve x^{k+1}&=\text{argmin}_{\ve x} \frac{1}{2}\ve x^{T}\ve H\ve x+\ve h^T\ve x+\ve y^{kT}(\ve x-\ve z^k)\nonumber\\
&+\frac{t^k}{2}\norm{\ve x-\ve z^k}_2^2\label{eq:ADMM}\\
\ve z^{k+1}&=\text{argmin}_{\ve z} g(\ve z)+\ve y^{kT}(\ve x^{k+1}-\ve z)+\frac{t^k}{2}\norm{\ve x^{k+1}-\ve z}_2^2\nonumber\\
\ve y^{k+1}&=\ve y^k+t^k(\ve x^{k+1}-\ve z^{k+1})\nonumber
\end{align}
where $\ve y$ denotes the Langrangian multipler. To solve the first iteration in (\ref{eq:ADMM}) using first order methods, we need to compute $\ve H\ve x^k$ for each steps.

From the above expressions,  it can be argued that the matrix-vector multiplications $\ve H \ve  x^k$ is (one of) the computationally most expensive operations\footnote{Indeed, for example the proximal operator $\prox$ could be very simple for many problems of interests.} in this algorithm for each step $k$.  With the focus on the matrix-vector multiplication, we will denote the update rule in (\ref{eq:proximal}) or (\ref{eq:ADMM}) simply as $\ve x^{k+1}=f(\ve H\ve x^k)$ in the sequel.

%\footnote{The evaluation of the proximal operator should also be taken into account. However, this operation could be very simple for many problems of interests. We will show some examples in Section \ref{sec:application}.}
In modern large-scale machine learning problems, the matrix $\ve H$ can be very large so that  computing the matrix multiplication (or even storing the matrix) in one processor is not feasible. In order to handle such large-scale problems,  we turn to a distributed computation paradigm where the task of computing $\ve H\ve x^k$ is collaboratively accomplished by several processors.   As discussed in Introduction,  the uncertainty (latency for example) of the individual  processor could be detrimental to the distributed computation system and renders the distributed computation approach unusable. To alleviate this problem, previous works (e.g. \cite{lee_speeding_2015} \cite{ferdinand_anytime_2016} \cite{tandon_gradient_2016}) have proposed coded computation schemes by adding redundancy in the computation tasks. We give a very simple example to illustrate the idea of coded computation. The matrix $\ve H$ is vertically split into two smaller matrices $\ve H_1, \ve H_2$. We use three processors to store $\ve H_1,\ve H_2, \ve H_1+\ve H_2$ separately, and each processor performs a smaller matrix-vector multiplication. It is easy to see that with \textit{any} two of three multiplicaitons  $\ve H_1\ve x, \ve H_2\ve x, (\ve H_1+\ve H_2)\ve x$, the user is able to recover $\ve H\ve x$. The same idea can be applied to a general setting with more users.

In this work, we take a step further to combine a coded sequential computation scheme with a modified algorithm. In particular, we propose the sequential approximation algorithm  shown in Algorithm \ref{al:sequential_approximation} for solving the problem in (\ref{eq:objective}).  In contrast to the original algorithm, Algorithm \ref{al:sequential_approximation} executes a \textit{sequence} of approximated problems (called ``approximations"). Each approximation is in the same form of the original problem, but with a different choice of the matrix $\ve H_{(r)}$. Notice that in order to obtain the correct solution in the end, the last approximation matrix $\ve H_{(R)}$ should be equal to $\ve H$. Algorithm \ref{al:sequential_approximation} should possess the following two properties to be useful:
\begin{itemize}
\item 1) Executing the approximations $f(\ve H_{(r)}\ve x^k)$ should be faster than executing the original iteration $f(\ve H\ve x^k)$ in a distributed system.
\item 2) The matrix $\ve H_{(r)}$ approaches the original matrix $\ve H$ as $r$ increases.
\end{itemize}

Property 1) ensures the proposed algorithm is faster  than the original algorithm in the approximation phases, and  Property 2) guarantees that Algorithm \ref{al:sequential_approximation} eventually provides a solution which is close or identical to the solution of the original problem. 

\begin{algorithm}
\begin{algorithmic} 
\Require $\ve H_{(1)},\ldots, \ve H_{(R)}$ are $R$ matrices which approximate $\ve H$ with an increasing accuracy. 
 \For {$k=1,\ldots, T_1$}  
\State $\ve x^{k+1}=f(\ve H_{(1)}\ve x^k)$ 
\EndFor \Comment{first approx. for $T_1$ iterations}
\For {$k=T_1+1,\ldots, T_1+T_2$}
\State $\ve x^{k+1}=f(\ve H_{(2)}\ve x^k)$ 
\EndFor \Comment{second approx. for $T_2$ iterations}

\quad \vdots

\quad \vdots

\For {$k=\sum_{r=1}^{R-1}T_{r}+1,\ldots, \sum_{r=1}^{R} T_r$}
\State $\ve x^{k+1}=f(\ve H_{(R)}\ve x^k)$ 
\EndFor
\end{algorithmic} 
\caption{A sequential approximation  for solving (\ref{eq:objective}). The function $f(\cdot)$ represents the updating rule in (\ref{eq:proximal}) or (\ref{eq:ADMM})} 
\label{al:sequential_approximation}
\end{algorithm}

An illustration of the sequential approximation approach is given in Fig. \ref{fig:path}. The black trajectory on the botten \footnote{Notice that the actual trajectory of $\ve x^k$ is not necessarily a straight line in $\mathbb R^n$. The plot is only an illustration.} represents the path of $\ve x^k$ when the variable is updated using the exact computation $\ve H\ve x$, and the  colored ``detour" represents the trajectory using the sequential approximation method.

\begin{remark}
It is natural to ask if the sequential approximation method in Algorithm \ref{al:sequential_approximation} is already useful in a system without stragglers, where computing $\ve H_{(r)}\ve x$ takes the same amount of time as computing $\ve H\ve x$. A preliminary investigation suggests that it will depend on both the algorithm and the optimization problem (e.g. the condition number of the matrix). For certain problems, the sequential approximation approach can indeed provide a better convergence rate even for systems without stragglers. The results will be reported in our future work.
\end{remark}

\begin{figure}
\centering
\includegraphics[scale=0.33]{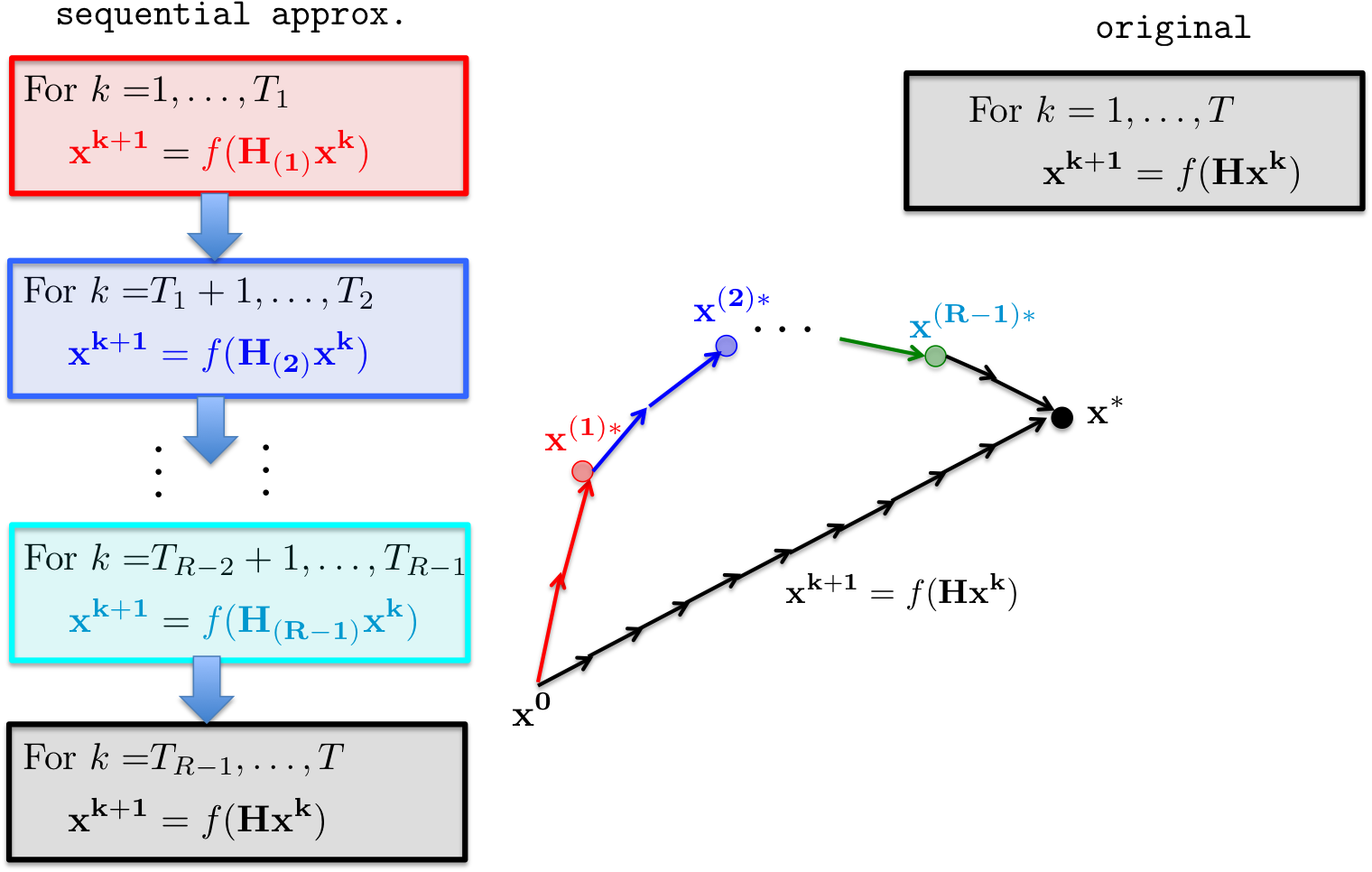}
\caption{An illustration of the sequential approximation approach. The black path on the botten represents the path of $\ve x^k$ when the variable is updated using the exact computation $\ve H\ve x$, and the  colored ``detour" represents the trajectory using sequential approximation. During the $r$-th approximation phase, the variable $\ve x^k$ converges to the point $x^{(r)*}$, which is the optimal solution to the optimization problem (\ref{eq:objective}) with $\ve H$ replaced by $\ve H_{(r)}$. The sequential approximation approach could be faster in a distributed computation system if each iteration in the ``detour" takes less computation time.}
\label{fig:path}
\end{figure}

\section{Coding for Distributed Sequential Matrix-Vector Multiplication }
\label{sec:coding}
A coded sequential computation scheme for the distributed system is the key mechanism behind the sequential approximation method.  In this section, we formally introduce the coded sequential computation problem. It can be viewed as a general problem formulation of the anytime coding scheme studied in \cite{ferdinand_anytime_2016}.

Consider a system with $L$ processors where each processor  performs a matrix-vector multiplication of the form $\ve B\ve z$. The matrix  $\ve B$ is of dimension $n\times m$ and $\ve z$ is a $m$-length vector. Let $\ve A_1,\ldots, \ve A_L$ be $L$ matrices prescribed by the user   
where the matrix $\ve A_i$ is of dimension $k_i\times m$.  The goal is to compute the matrix-vector multiplications $\ve A_i\ve z, i=1,\ldots, L$ for a vector $\ve z$ using these $L$ processors. To perform the computation in a distributed manner, $L$ matrices $\ve B_i\in\mathbb R^{n\times m}, i=1,\ldots, L$ are generated based on the given matrices $\ve A_i\in\mathbb R^{k_i\times m}, i=1,\ldots L$, namely
\begin{align}
\ve B_i=\mathcal E_i(\ve A_1,\ldots, \ve A_L)
\label{eq:encoder}
\end{align}
where $\mathcal E_i$ denotes a mapping $\prod_j\mathbb R^{k_j\times m}\mapsto \mathbb R^{n\times m}$. The matrix $\ve B_i$ is stored in the $i$-th  processor. To compute the multiplication, a vector $\ve z$ is given to all processors, and each processor returns the result $
\ve y_i=\ve B_i\ve z$ to the user when it finishes  the computation. The user then applies a decoder $\mathcal D$ to obtain desired multiplications $\ve A_i\ve z$ using the received results $\ve y_i$. We would like to design our encoders $\mathcal E_i$ and $\mathcal D$ such that the system has the following property.
\begin{property}[Sequential computation]
With the computation results from \textit{any} $\ell$ processors ($1\leq \ell\leq L$), $\ve y^{(1)}:=\ve B_{i_1}\ve z, \ldots, \ve y^{(\ell)}:=\ve B_{i_\ell}\ve z$ for some $i_1,\ldots,i_\ell\in[L]$, the user can recover $\ve A_1\ve z, \ldots, \ve A_\ell \ve z$ where $\ve A_i\in\mathbb R^{k_i\times m}$.
\label{property:sequential}
\end{property}

\begin{figure}[htb]
\centering
\includegraphics[scale=0.5]{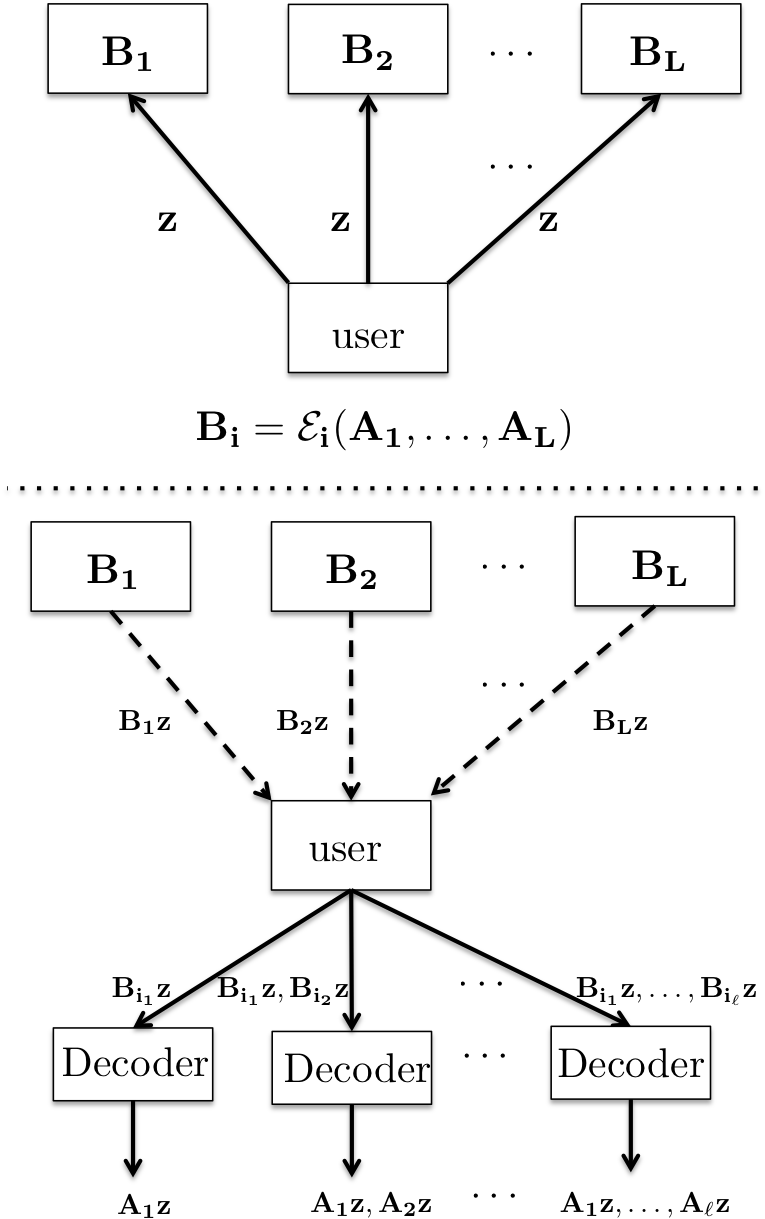}
\caption{An illustration for distributed sequential matrix-vector multiplication.  $L$ matrices $\ve B_1,\ldots, \ve B_L$ are generated based on the prescribed matrices $\ve A_1,\ldots, \ve A_L$. The user first sends a vector $\ve z$ to all processors (top). Whenever a processor returns the computation result $\ve B_{i_\ell}\ve z$, the user can recover an additional matrix-vector multiplication $\ve A_\ell \ve z$ (bottom), $\ell = 1,\ldots, L$. }
\label{fig:system}
\end{figure}

Figure \ref{fig:system} gives an illustration of the sequential computation scheme of the distributed computation system. If a coding scheme satisifes the above sequential computation property, the corresponding value $(k_1,\ldots, k_L)$ is called a  \textit{feasible configuration}. The question is that given a distributed computation system with parameters $(L, n)$\footnote{As we shall see, the parameter $m$ can be chosen arbitrarily.}, what are the feasible configurations (possible values of $k_1,\ldots, k_L$), and how do we design the encoders $\mathcal E_i$ and the decoder $\mathcal D$ for such a system.

This sequential computation scheme is useful for distributed computation systems with stragglers because it guarantees that any finished processor will provide  useful results. Moreover, we could choose the matrices such that $\ve A_i\ve z$ is more crucial than $\ve A_j\ve z$ for our application if $i\leq j$, such that  ``more important" results are received earlier. As pointed out in \cite{ferdinand_anytime_2016}, this coding scheme can also be viewed as an approximation method for computing the multiplication $\ve A\ve z$, where the accuracy increases gradually as more and more processors finish their tasks.

\subsection{The coding scheme}
\label{sec:coding_scheme}

In this section, we give a coding scheme for the sequential distributed matrix-vector multiplication problem.   This scheme is a generalization of the MDS codes based scheme used in \cite{lee_speeding_2015}, and uses essentially the same idea as in \cite{puri_multiple_1999} (multiple description coding) and \cite{ferdinand_anytime_2016}. %Here we state it for the general case in our problem setting.

\textbf{Coding scheme:}
For each $i$, we divid the matrix $\ve A_i$ vertically into at most $\lfloor k_i/i\rfloor+1$ submatrices as follows
\begin{align*}
&\ve A_i^{(1)}:=
\begin{pmatrix}
\ve A_{i,1}\\
\vdots\\
\ve A_{i,i}
\end{pmatrix},
\ve A_i^{(2)}:=
\begin{pmatrix}
\ve A_{i,i+1}\\
\vdots\\
\ve A_{i,2i}
\end{pmatrix},
\ldots,\\
&\ve A_i^{ (\lfloor k_i/i \rfloor)} :=
\begin{pmatrix}
\ve A_{i,\lfloor k_i/i\rfloor (i-1) +1}\\
\vdots\\
\ve  A_{i,\lfloor k_i/i\rfloor i}
\end{pmatrix},\\
&\ve A_i^{(\lfloor k_i/i \rfloor+1)}:=
\begin{pmatrix}
\ve A_{i,\lfloor k_i/i\rfloor +1}\\
\vdots\\
A_{i,\lfloor k_i/i\rfloor +mod(k_i,i)}
\end{pmatrix}
\end{align*}
where $\ve A_{i,j}$ denotes the $j$-th row of the matrix $\ve A_i$. In the case when $i$ divides $k_i$, we do not have the last matrix $\ve A_i^{(\lfloor k_i/i \rfloor+1)}$.

For each matrix $\ve A_i^{(j)}\in\mathbb R^{i\times m}$ where  $j=1,\ldots, \lfloor k_i/i \rfloor$, we encode its rows to form a new matrix $\ve {\tilde A}_i^{(j)}\in\mathbb R^{L\times m}$. Particularly, we use a systematic MDS code, such that the first $i$ rows of $\ve {\tilde A}_i^{(j)}$ is identical to $\ve A_i^{(j)}$, and the last $L-i$ rows are linear combinations of the rows of $\ve A_i^{(j)}$ (i.e., parity checks).  If $i$ does not divide $k_i$, the rows of the last matrix $\ve A_i^{(\lfloor k_i/i \rfloor+1)}\in\mathbb R^{mod(k_i,i)\times m}$ is encoded with an $(L-i+mod(k_i,i), mod(k_i,i))$ MDS code into a new matrix $\ve{\tilde A}_i^{\lfloor k_i/i \rfloor+1}\in\mathbb R^{(L-i+mod(k_i,i))\times m}$.

The matrices $\ve B_i$ are generated using the encoded matrices $\ve {\tilde A}_i^{(j)}$. More precisely, each matrix $\ve B_i$ contains exactly one row of the matrix $\ve{\tilde A}_i^{(j)}$, for $j=1,\ldots, \lfloor k_i/i \rfloor$ and for all $i=1,\ldots, L$. If we have the extra matrix $\ve{\tilde A}_i^{\lfloor k_i/i \rfloor+1}$ for certain $i$, its $L-i+mod(k_i,i)$ rows are distributed to arbitrary $L-i+mod(k_i,i)$ matrices $\ve B_i$.

\textbf{Example:} We apply the above coding scheme to a system with parameters $L=4, n=3$, and the configuration $(k_1=0, k_2=3, k_3=3, k_4=1)$. There are three matrices $\ve A_2\in\mathbb R^{3\times m}, \ve A_3\in\mathbb R^{3\times m}$ and $\ve A_4\in\mathbb R^{1\times m}$ to encode ($\ve A_1$ is zero in this case).  The proposed coding scheme generates matrices $\ve B_i$ as follows:
%\begin{scriptsize}
\begin{align*}
&\ve B_1=\begin{pmatrix}
{\color{black}\ve A_{21}}\\
{\color{black}\ve A_{23}}\\
{\color{black}\ve A_{32}}
\end{pmatrix}, 
\ve B_2=\begin{pmatrix}
{\color{black}\ve A_{22}}\\
{\color{black}\ve A_{23}}\\
{\color{black}\ve A_{33}}
\end{pmatrix}, 
\ve B_3=\begin{pmatrix}
{\color{black}\ve A_{21}+\ve A_{22}}\\
{\color{black} \ve A_{23}}\\
{\color{black}\ve A_{31}+\ve A_{32}+\ve A_{33}}
\end{pmatrix}, \\
&\ve B_4=\begin{pmatrix}
{\color{black} \ve A_{21}+\ve A_{22}}\\
{\color{black} \ve A_{31}}\\
{\color{black}\ve A_4}
\end{pmatrix}.
\end{align*}%\end{scriptsize}
It can be checked that we can recover $\ve A_2\ve z$  if we have $\ve B_i\ve z, \ve B_j\ve z$ for any $i,j$,  recover $\ve A_3\ve z$ if we have $\ve B_i\ve z, \ve B_j\ve z, \ve B_k\ve z$ for any $i, j, k$ and recover $\ve A_4\ve z$  with all the compuation results.

%The following theorem states that the coding scheme described above gives the desired sequential computation property.
\begin{theorem}[Coding scheme]
Consider the distributed sequential matrix-vector multiplication problem with parameters $(n,L)$. The configuration $(k_1,\ldots, k_L)$ is feasible if it satisfies
\begin{align}
\sum_{i=1}^L s_i\leq nL
\label{eq:achievable}
\end{align}
where $s_i$ is defined as
\begin{align}
s_i:=
\begin{cases}
\frac{k_i}{i}L \quad &\text{$i$ divides $k_i$}\\
\lfloor\frac{k_i}{i}\rfloor \cdot L+ L-i+mod (k_i,i)\quad &\text{otherwise}
\end{cases}
\label{eq:si}
\end{align}
\label{thm:coding_scheme}
\end{theorem}
\begin{proof}
The proof is given in Appendix.
\end{proof}

\begin{remark}
We point out that this result is a generalization of the coding scheme using a single MDS code proposed in \cite{lee_speeding_2015}, which can be seen as a special configuration with $(k_1=0,\ldots, k_{\ell-1}=0, k_\ell=\ell n, k_{\ell+1}=0,\ldots, k_L=0)$ for some $\ell$.
\end{remark}

\begin{remark}  [Complexity of decoding] With the computation results from $\ell$ processors, the decoding process at the user is equivalent to solving a linear system of at most $n\ell$  unknowns, which does not depend on $m$ (number of columns of the matrices). Hence this coding scheme is most beneficial for computing matrix-vector multiplications when the number of rows of matrices is very large. Moreover, using MDS codes with special structures (Reed-Solomon codes for example), the decoding process is often much simpler than solving a generic linear system. % with $n\ell$ unknowns.
\end{remark}

We can show that if we restrict  ourselves to linear coding schemes (i.e., the encoder in (\ref{eq:encoder}) is a linear function of $\ve A_1,\ldots, \ve A_L$), then the coding method in Section \ref{sec:coding_scheme} is the best possible. This result establishes the optimality of our coding scheme in the previous section.

\begin{theorem}[Converse for linear  schemes] Consider the distributed sequential matrix multiplication problem with parameters $(n,L)$. Under linear coding schemes, any feasible configuration $(k_1,\ldots, k_L)$ must satisfy the constraint (\ref{eq:achievable}).
\label{thm:converse}
\end{theorem}
\begin{proof}
The proof is given in  Appendix.
\end{proof}

\section{The Sequential Approximation Method: Examples}
\label{sec:application}
Equipped with the coded sequential computation scheme described in Section \ref{sec:coding}, the sequential approximation method in Algorithm \ref{al:sequential_approximation} can be implemented where the matrix-vector multiplication $\ve H_{(r)}\ve x^k$ is computed in  a sequential manner. We demonstrate this method by considering the  Lasso problem in its standard form
\begin{align}
\text{minimize } \frac{1}{2}\norm{\ve F\ve x-\ve b}_2^2+\gamma\norm{\ve x}_1
\label{eq:lasso}
\end{align}
for a matrix $\ve F\in\mathbb R^{w\times m}$ and $\gamma\geq 0$.  This corresponds to the optimization problem in (\ref{eq:objective}) by identifying $\ve H=\ve F^T\ve F, \ve h = \ve F^T\ve b$ and $g(\ve x) = \gamma \norm{\ve x}_1$. The corresponding proximal gradient method for this problem is given by
\begin{align}
\ve x^{k+1}=S_\gamma(\ve x^{k} - t^k(\ve F^T\ve F\ve x^k-\ve F^T\ve b))
\label{eq:proximal_lasso}
\end{align}
%\begin{align*}
%x^{k+1}&:=(A^TA+\rho I)^{-1}(A^Tb+\rho(z^k-u^k))\\
%z^{k+1}&:=S_{\lambda/\rho}(x^{k+1}+u^k)\\
%u^{k+1}&:=u^k+x^{k+1}-z^{k+1}
%\end{align*}
where $S_\gamma(x)$ is the \textit{soft-thresholding operator} defined as
\begin{align*}
S_\gamma(x)=
\begin{cases}
x_i-\gamma &x_i\geq \gamma\\
0 &|x_i|\leq \gamma\\
x_i+\gamma &x_i\leq -\gamma
\end{cases}.
\end{align*}
For this problem, the soft-thresholding operator is very simple, and the most computationally expensive step in the algorithm is the matrix-vector multiplication $\ve F^T\ve F  \ve x^k$ for each step $k$.

\subsection{Approximations}

Instead of computing $\ve F^T \ve F\ve x^k$  at each step $k$, we use the proposed sequential distributed computation scheme. Similar to \cite{lee_speeding_2015},  we frist focus on computing the term $\ve F\ve x^k$. Using Algorithm \ref{al:sequential_approximation}  to solve  the above problem requires a specification of $\ve F_{(r)}$. A priori, $\ve F_{(r)}$ could be chosen in any way, as long as the multiplication $\ve F_{(r)}\ve x^k$ requires less computation than $\ve F\ve x^k$. Similar to \cite{ferdinand_anytime_2016}, in this paper we  choose $\ve F_{(r)}$ to be a low rank approximation of $\ve F$ using singular value decomposition:
\begin{align*}
\ve F=\ve U\Sigma \ve V^T=\sum_{i=1}^d\sigma_i \ve u_i \ve v_i^T
\end{align*}
where $\sigma_1\geq \sigma_2\geq \ldots\geq\sigma_d$ denote the singular values of $\ve F$ with rank $d$.   $\ve u_i\in\mathbb R^w, \ve v_i\in\mathbb R^m$ are the $i$-th column of $\ve U$ and $\ve V$, respectively. In particular, we  choose $\ve F_{(r)}$ as
\begin{align}
\ve F_{(r)}:=\sum_{i=1}^{\text{rank}(
\ve F_{(r)})}\sigma_i \ve u_i\ve v_i^T
\label{eq:A_r}
\end{align}
for some $1\leq \text{rank}(
\ve F_{(r)})\leq d$. Namely $\ve F_{(r)}$  captures the largest $\text{rank}(
\ve F_{(r)})$ singular values  of $\ve F$.

%To implement Algorithm \ref{al:sequential_approximation} using the proposed sequential computation scheme, the encoding and decoding processes are given as follows.  
Define $\ve A_1:=[\ve v_1^T;  \ldots ;\ve v_{k_1}^T]\in\mathbb R^{k_1\times m}, \ve A_2:=[\ve v_{k_1+1}^T;  \ldots ;\ve v_{k_1+k_2}^T]\in\mathbb R^{k_2\times m}, \ldots, \ve A_L:=[\ve v_{d-k_L+1}^T; \ldots ; \ve v_d^T]\in\mathbb R^{k_L\times m}$ for a chosen configuration $(k_1,\ldots, k_L)$ which satisfies the condition (\ref{eq:achievable}).  Using the coding scheme described in Section \ref{sec:coding_scheme}, we generate $L$ matrices $\ve B_1,\ldots, \ve B_L \in\mathbb R^{n\times m}$ based on $\ve A_1,\ldots, \ve A_L$ for the $L$ processors.

\begin{algorithm}
\caption{Subroutine for computing $\ve F_{(r)}\ve x^k$: processor $i$} 
\begin{algorithmic}[1]
 
\Require $\ve B_i$ is the encoded matrix stored in processor $i$
\State Receive an input vector $\ve x^k$  
\State Compute $\ve y_i:=\ve B_i\ve x^k$
\State Send $\ve y_i$ back to the user
\end{algorithmic}
\label{al:subroutine_processor}
\end{algorithm}

\begin{algorithm}
\caption{Subroutine for computing $\ve F_{(r)}\ve x^k$: user}
\begin{algorithmic}[1]
 
%\Require when calculating with the approximation $\ve A_{(r)}$

\State Send a vector $\ve x^k$ to all processors
\State Wait until $\ell$ processors finish, where $\sum_{i=1}^\ell k_i\geq \text{rank}(\ve F_{(r)})$
\State Decode $\ve F_{(r)}\ve x^k$ using $\ve y_{i_1},\ldots, \ve y_{i_\ell}$ from $\ell$ processors
%\State Perform a multiplication to obtain $\ve H_{(r)}\ve x^k$ as in (\ref{eq:recover}). 
\end{algorithmic}
\label{al:subroutine_master}
\end{algorithm}

When executing the algorithm, the vector $\ve x^k$ is given to all processors at each time step $k$. The coding scheme guarantees that with the computation results $\ve y_i:=\ve B_i \ve x^k$ from  any $\ell$ processors $(1\leq \ell\leq L)$, the user can  recover the multiplication result
\begin{align}
\ve t_\ell:=\begin{pmatrix}
\ve A_1\\
\vdots\\
\ve A_\ell
\end{pmatrix}\ve x^k =\begin{pmatrix}
\ve v_1^T\\
\vdots\\
\ve v_{h_\ell}^T
\end{pmatrix}\ve x^k \in\mathbb R^{h_\ell}
\label{eq:t_l}
\end{align}
where we define $h_\ell:=\sum_{i=1}^\ell k_i$. An extra multiplication gives $\ve F_{(r)}\ve x^k$:
\begin{align}
[\sigma_1\ve u_1,\ldots, \sigma_{h_\ell}\ve u_{h_\ell}]\ve t_\ell=\sum_{i=1}^{h_\ell} \sigma_i \ve u_i\ve v_i^T\ve x^k=\ve F_{(r)}\ve x^k
\label{eq:recover}
\end{align}
where  $\text{rank}(\ve F_{(r)})=h_\ell$ in this case.  We point out that in the lasso problem, the matrix $\ve F$ in general has much more columns than rows, hence the above multiplication (\ref{eq:recover}) takes less computation. Moreover, it can also be done using the distributed system in a similar way. By treating $\ve F_{(r)}\ve x^k$ as the vector to be multiplied, the user can distribute the multiplication $\ve F^T_{(r)}\ve F_{(r)}\ve x^k$ in the same way. Hence there are two computation steps for each iteration. We omit the details of the second step in this paper.

The subroutines for processors and for the user are given in Algorithm \ref{al:subroutine_processor} and Algorithm \ref{al:subroutine_master}, respectively. We point out that the subroutine for processors does not change for different approximation level (different $\ve F_{(r)}$),  and  only the user needs to adjust its procedure to adapt to different approximation levels.

\subsection{Computation time}
\label{subsec:rational}

As mentioned in Section \ref{sec:overview},  the reason for adopting the sequential approximation method is that it is faster to obtain a low rank approximation $\ve F_{(r)}\ve x^k$ than obtaining the exact answer $\ve F\ve x^k$ in a distributed system with stragglers. More precisely, let $T_i$ denote the random computation time of processor $i$ and let $T_{(\ell)}$ denote the $\ell$-th order statistic, i.e.
\begin{align*}
T_{(\ell)}:=\text{$\ell$-th smallest of }\{T_1,\ldots, T_L\}.
\end{align*}
The time for recovering the result $\ve F_{(r)}\ve x^k$ is given by $T_{(\ell)}$ where $\ell$ satisfies
\begin{align*}
\sum_{i=1}^\ell k_i\geq \text{rank}(\ve H_{(r)}),
\end{align*}
while computing the exact answer $\ve F\ve x^k$ requires time $T_{(\ell')}$ where $\ell'$ satisfies
\begin{align*}
\sum_{i=1}^{\ell'} k_i\geq \text{rank}(\ve F).
\end{align*}
Notice that we always have $\ell \leq \ell'$.  In a distributed system with stragglers, $T_{(\ell')}$ could be significantly larger than $T_{(\ell)}$ if $\ell'$ is larger than $\ell$.

In summary, although  Algorithm \ref{al:sequential_approximation} starts with ``incorrect" iterations with approximation matrices $\ve F_{(r)}$, the computation time for those iterations are shorter than using the exact matrix $\ve F$. If we choose approximations $\ve F_{(r)}$ judiciously (approaching $\ve F$ gradually), the variable $\ve x^k$ will approach the optimal solution, but with a  shorter computation time.

\subsection{Choices of parameters}
There are many free parameters to choose for the sequential approximation algorithm, including the approximation matrices $\ve F_{(r)}$, the number of different approximation levels $R$, and the number of iterations $T_r, r=1,\ldots, R$.  They should be chosen in a way such that the algorithm can be implemented with the given system parameters $(L,n)$. In other words, if we choose $\ve F_{(r)}$ as in (\ref{eq:A_r}), there should exist a feasible configuration $(k_1,\ldots, k_L)$ which both satisfy  condition  (\ref{eq:achievable}) and the condition
\begin{align}
\sum_{i=1}^{\ell(r)} k_i\geq \text{rank}(\ve F_{(r)})
\label{eq:feasible_rank}
\end{align}
for every $r=1,\ldots, R$ with some $\ell(r)\in [1:L]$ satisfying $\ell(r_1)\geq \ell(r_2)$ if $r_1\geq r_2$. Recall that these choices will affect the running time of each iteration as shown in Subsection \ref{subsec:rational}, hence we could optimize the overall execution time with respect to these parameters.

More importantly, choosing $\ve F_{(r)}$ as (\ref{eq:A_r}) is only one way to approximate the matrix $\ve F$. There may exist other choices of $\ve F_{(r)}$ which provide a better convergence performance. The problem of choosing good approximations will be addressed in the future work.

\subsection{Numerical results}

In this section, we present  numerical results to demonstrate the sequential approximation method with the Lasso problem in (\ref{eq:lasso}). Specifically, we consider the  proximal gradient algorithm in (\ref{eq:proximal_lasso}). %and with the sequential approximation algorithm in Algorithm \ref{al:sequential_approximation}.

\begin{figure}[!h]
\centering
\includegraphics[scale=0.45]{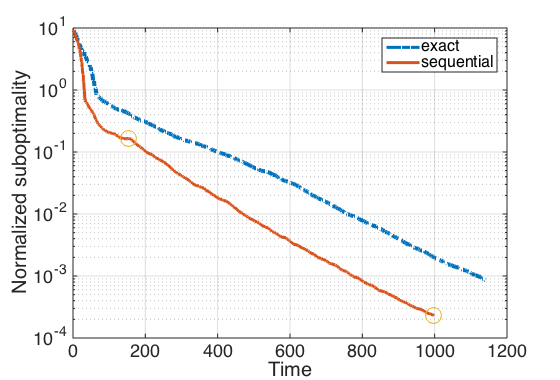}
\caption{Numerical results for Example 1, showing the normalized suboptimality $\norm{x^k-x^*}/\norm{x^*}$ verses overall computation time.  The solid line denotes the sequential approximation algorithm and the dashed line denotes the original proximal gradient method.   A black circles denotes the end of one approximation level. To reach a normalized suboptimality $10^{-3}$, the sequential approximation is roughly $30\%$ faster than the original algorithm. }
\label{fig:ex_1}
\end{figure}

As a toy example, we assume that the distributed system has $4$ processors where each processor can compute a matrix-vector multiplication $\ve B\ve z$ with $\ve B\in\mathbb R^{10\times m}$, namely  $(L=4, n=10)$. The $\ve F$ matrix in the Lasso problem has dimension $38\times m$  with rank $38$ and $m=500$.  In our simulations, the matrix $\ve F$ is chosen randomly, and the regularization coefficient $\gamma$ is chosen to be $5$. The computation time of each processor is assumed to have an exponential distribution with the density function $p(x) = \lambda \exp(-\lambda x)$ with the choice $\lambda=1$.

\begin{figure}[!htb]
\centering
\includegraphics[scale=0.45]{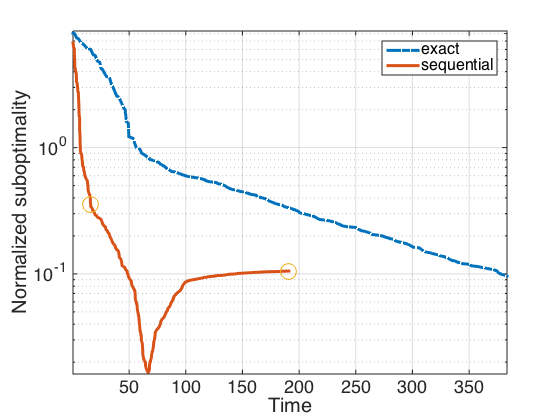}
\caption{Numerical results for Example 2, showing the normalized suboptimality $\norm{x^k-x^*}/\norm{x^*}$ verses overall computation time.   The solid line denotes the sequential approximation algorithm and the dashed line denotes the original proximal gradient method. In this example, the algoritm never uses the exact matrix $\ve F$, hence $\ve x^k$ does not converge to the  optimal solution of the original proble, but an approximate soultion with a suboptimality $0.1$.  In this case, the sequential approximation algorithm is almost two times faster than the original algorithm to reach this suboptimality.}
\label{fig:ex_2}
\end{figure}

\textit{Example 1:} We use two approximation levels where $\ve F_{(1)}, \ve F_{(2)}$ are chosen according to (\ref{eq:A_r}) with $\text{rank}(\ve F_{(1)})=6$ and $\text{rank}(\ve F_{(2)})= 38$ (i.e., $\ve F_{(2)}=\ve F$). It can be checked straightforwardly this particular choice can be implemented with a feasible configuration $(k_1=k_2=0, k_3=6, k_4=32)$.  With this choice, the matrix-vector multiplication $\ve F_{(1)}\ve x^k$ is obtained when three processors finish their tasks. It can be seen that if the user needs the exact result $\ve F\ve x^k$, he must wait for all four processors to finish. It can be calculated that the average waiting time for three processors is approximately $1.08$, while the averate waiting time for four processors is $2.08$.  The  performance of the algorithm is given in Figure \ref{fig:ex_1}. 

\textit{Example 2 (approximate solution):} This approach is also useful when we only aim to obtain an approximate solution of the problem. In this example we execute the sequential approximation algorithm with  two levels of approximation where $\text{rank}(\ve F_{(1)})=5$ and $\text{rank}(\ve F_{(2)})= 15$.  This is implemented with a configuration $(k_1=5, k_2=10, k_3=0, k_4=0)$.  In other words, the user can recover $\ve F_{(1)}\ve x^k$ when any processor returns the result and can recover $\ve F_{(2)}\ve x^k$ when any $2$ processors return the results.  The average waiting time for one processor is approximately $0.25$ and the average waiting time for two processors is approximately $0.58$. Notice that in this case, the sequential approximation algorithm cannot converge to the true minimizer since it does not use the exact matrix $\ve F$ in any level.   Nevertheless the simulation results   show that it gives a fairly good approximate solution.   The  performance of the algorithm is given in Figure \ref{fig:ex_2}.

\section{Appendix}

%\subsubsection*{Proof of Theorem \ref{thm:coding_scheme}}
\begin{proof}[Proof  of Theorem \ref{thm:coding_scheme}]
It is easy to see that each matrix $\ve A_i$ is encoded into at most $\lfloor k_i/i\rfloor +1$ matrices $\tilde {\ve A}^{(1)},\ldots,\tilde {\ve A}^{(\lfloor k_i/i\rfloor +1)}$, whose total number of rows sums up to $s_i$ as defined in (\ref{eq:si}). Hence all the encoded data can be accommodated in the $L$ matrices if it satisfies $\sum_{i=1}^Ls_i\leq nL$.

Now we argue that this coding scheme possesses the sequential  property in Property \ref{property:sequential}. Since the rows of $\ve{A}_i^{(j)}, j=1,\ldots, \lfloor k_i/i\rfloor$ are encoded with an $(L,i)$ MDS code into the matrix $\tilde {\ve A}_i^{(j)}$, it is to see that any $i$ entries of the matrix-vector multiplication $\tilde{\ve A}_i^{(j)}\ve x$ allow us to recover $\ve A_i^{(j)}\ve x$. Recall that each row of $\tilde {\ve A}_i^{(j)}$ is distribtued to one processor, hence if $i$ processors return the their matrix-multiplication results, the user can recover $\ve A_i^{(j)}\ve x$, and this holds for all $j=1,\ldots, \lfloor k_i/i\rfloor$. 

If $i$ does not divide $k_i$,  we have the extra matrix $\ve A_i^{(\lfloor k_i/i\rfloor+1)}$ whose rows are encoded using an $((L-i+mod(k_i,i), mod(k_i,i))$ MDS codes into $\tilde{\ve A}_i^{(\lfloor k_i/i\rfloor+1)}$ with $(L-i+mod(k_i,i)$ rows, and each row is distrbued to one processor. It can be seen that any $i$ processors will contain at least $mod(k_i,i)$ rows of $\tilde{\ve A}_i^{(\lfloor k_i/i\rfloor+1)}$. Indeed, there are $i-mod(k_i,i)$ processors which do not store any row of $\tilde{\ve A}_i^{(\lfloor k_i/i\rfloor+1)}$. The worst case is when a subset of  $i$ processors include all the processors which do not contain rows of $\tilde{\ve A}_i^{(\lfloor k_i/i\rfloor+1)}$. Even in this case, we have the remaining $mod(k_i,i)$ processors which contain rows of $\tilde{\ve A}_i^{(\lfloor k_i/i\rfloor+1)}$.  This shows that it is always possible to recover $\ve A_i^{(j)}\ve x, j=1,\ldots, \lfloor k_i/i\rfloor$ and $ \ve A_i^{(\lfloor k_i/i\rfloor+1)}\ve x$ (hence $\ve A_i\ve x$) with any $i$ processors. This argument holds for all $i, i=1,\ldots, L$ and thus concludes the proof.
\end{proof}

\begin{proof}[Proof sketch of Theorem \ref{thm:converse}]
We will only consider linear coding scheme in this paper where the matrix $\ve B_i$ is a linear function of matrices $\ve A_1,\ldots, \ve A_L$. It can be argued  that, any linear coding scheme can be reduced to a scheme where each row of $\ve B_i$  only consists of rows of one matrix $\ve A_i$ (as shown in the Example in Section \ref{sec:coding_scheme}). In other words, any linear scheme  can be equivalently implemented as
\begin{align}
\ve B_i=\ve G_i\begin{pmatrix}
\ve A_1\\
\vdots\\
\ve A_L
\end{pmatrix}
\label{eq:linear_scheme}
\end{align}
where  $\ve G_i\in\mathbb R^{n\times \sum_{i=1}^L k_i}$ is  of the form
\begin{align*}
\ve G_i=\begin{pmatrix}
\ve G_{i,1} &\ve 0 &\ldots &\ve 0\\
 \ve 0& \ve G_{i,2} &\ldots &\ve 0\\
 \vdots &\vdots &\vdots &\vdots\\
 \ve 0 &\ve 0 &\ldots &\ve G_{i,L}
\end{pmatrix}
\end{align*}
for some $\ve G_{i,j}\in\mathbb R^{n_j\times k_i}$ where $\sum_{j=1}^Ln_j=n$. %A proof of this fact in given in Lemma \ref{lemma:bloack_diagonal} in Appendix.

Since each row of $\ve B_\ell$ is used for encoding only one matrix $\ve A_i$, we use $n_{i,\ell}$ to denote the number of rows used for encoding $\ve A_i$ in processor $\ell$ (which contains $\ve B_\ell$). Let $s_i$ be the number of rows used for the matrix $\ve A_i$ across all $L$ processors, i.e., $s_i=\sum_{\ell=1}^L n_{i,\ell}$.  Now we show that in order to have a valid coding scheme for the distributed storage problem, $s_i$ should satisfy
\begin{align*}
s_i\geq \begin{cases}
\frac{k_i}{i}L \quad &\text{$i$ divides $k_i$}\\
\lfloor\frac{k_i}{i}\rfloor \cdot L+ L-i+mod (k_i,i)\quad &\text{otherwise}
\end{cases}
\end{align*}
which matches the achievable coding scheme in (\ref{eq:si}).

First notice that in order to  recover $\ve A_i\ve z$ using any $i$ processors, a necessary condition on  $n_{i,\ell}$  is  
\begin{align}
\sum_{\ell\in \mathcal T}n_{i,\ell}\geq k_i
\end{align}
for any subset $\mathcal T\subseteq[L]$ with $|\mathcal T|=i$, simply because $\ve A_i\ve z$ has length $k_i$. Hence a lower bound $\hat s_i$ on $s_i$ is given by the following optimization problem
\begin{align}
\text{minimize }&\hat s_i:=\sum_{\ell=1}^L n_{i,\ell} \label{eq:optimize}\\
\text{s . t. }& \sum_{\ell\in \mathcal T}n_{i,\ell}\geq k_i\text{ for all } \mathcal T\subseteq[L] \text{ with }  |\mathcal T|=i \nonumber \\
&n_{i,l}\in\mathbb Z \text{ for all }\ell\in[L]
\end{align}

If we ignore the integer constraint on $n_{i,\ell}$, it is easy to see that the optimal solution to the relaxed linear programming problem is given by
\begin{align*}
n_{i,\ell}^*=k_i/i \text{ for all }\ell\in[L]
\end{align*}
and the lower bound $\hat s_i$ is equal to $Lk_i/i$. In the special case when $i$ divides $k_i$, the optimal solution $n_{i,\ell}^*$ is an integer hence is also an optimal solution to the original problem (\ref{eq:optimize}). This shows that we have $s_i\geq \hat s_i=Lk_i/i$ if $i$ divides $k_i$.

If $i$ does not divide $k_i$,  it can be argued that,  due to the complete symmetry of the problem (\ref{eq:optimize}), the optimal solution of the integer programming problem (\ref{eq:optimize}) satisfies
\begin{align*}
n_{i,\ell}^*=\lceil k_i/i \rceil \text{ or } n_{i,\ell}^*=\lfloor k_i/i\rfloor \text{ for all }\ell\in [L]
\end{align*}
Moreover, at most $\alpha:=i-mod(k_i,i)$ among the $L$ processors  are allowed to dedicate $\lfloor k_i/i \rfloor$ rows to $\ve a_i$, and all other processors must dedicate $\lceil k_i/i \rceil$ rows to $\ve A_i$.        Indeed, if  we have $n_{i,\ell}^*=\lfloor k_i/i\rfloor$ for all $\ell\in \mathcal T'\subseteq[L]$ for some set $\mathcal T'$ with  $|\mathcal T'|=\alpha+p$ for some $p\geq 1$, then for a set $\mathcal T\supseteq \mathcal T'$ we have
\begin{align*}
\sum_{\ell\in\mathcal T}n_{i,\ell}&=\sum_{\ell\in\mathcal T'}n_{i,\ell}+\sum_{\ell\in\mathcal T\backslash \mathcal T'}n_{i,\ell}\\
&=(\alpha+p)\lfloor k_i/i \rfloor+(i-(\alpha+p))(\lfloor k_i/i \rfloor+1)\\
&=i\lfloor k_i/i \rfloor+mod(k_i,i)-p\\
&<i\lfloor k_i/i \rfloor+mod(k_i,i)=k_i,
 \end{align*}
hence not be able to recover $\ve A_i\ve z$. We conclude that $\hat s_i$ must satisfy
\begin{align*}
\hat s_i= \sum_{\ell=1}^Ln_{i,\ell}&\geq \alpha \lfloor k_i/i \rfloor+(L-\alpha)(\lfloor k_i/i \rfloor+1)\\
&=L\lfloor k_i/i \rfloor+L-i+mod(k_i,i)
\end{align*}
This shows a lower bound on $s_i$ for the case when $i$ does not divides $k_i$. 

\end{proof}

\bibliographystyle{IEEEtran}
\bibliography{DistrComp}

\begin{thebibliography}{10}
\providecommand{\url}[1]{#1}
\csname url@samestyle\endcsname
\providecommand{\newblock}{\relax}
\providecommand{\bibinfo}[2]{#2}
\providecommand{\BIBentrySTDinterwordspacing}{\spaceskip=0pt\relax}
\providecommand{\BIBentryALTinterwordstretchfactor}{4}
\providecommand{\BIBentryALTinterwordspacing}{\spaceskip=\fontdimen2\font plus
\BIBentryALTinterwordstretchfactor\fontdimen3\font minus
  \fontdimen4\font\relax}
\providecommand{\BIBforeignlanguage}[2]{{%
\expandafter\ifx\csname l@#1\endcsname\relax
\typeout{** WARNING: IEEEtran.bst: No hyphenation pattern has been}%
\typeout{** loaded for the language `#1'. Using the pattern for}%
\typeout{** the default language instead.}%
\else
\language=\csname l@#1\endcsname
\fi
#2}}
\providecommand{\BIBdecl}{\relax}
\BIBdecl

\bibitem{dean_tail_2013}
J.~Dean and L.~A. Barroso, ``\BIBforeignlanguage{en}{The tail at scale},''
  \emph{\BIBforeignlanguage{en}{Communications of the ACM}}, vol.~56, no.~2,
  p.~74, Feb. 2013.

\bibitem{lee_speeding_2015}
\BIBentryALTinterwordspacing
K.~Lee, M.~Lam, R.~Pedarsani, D.~Papailiopoulos, and K.~Ramchandran, ``Speeding
  {Up} {Distributed} {Machine} {Learning} {Using} {Codes},''
  \emph{arXiv:1512.02673 [cs, math]}, Dec. 2015, arXiv: 1512.02673. [Online].
  Available: \url{http://arxiv.org/abs/1512.02673}
\BIBentrySTDinterwordspacing

\bibitem{ferdinand_anytime_2016}
N.~S. Ferdinand and S.~C. Draper, ``Anytime coding for distributed
  computation,'' in \emph{2016 54th {Annual} {Allerton} {Conference} on
  {Communication}, {Control}, and {Computing} ({Allerton})}, Sep. 2016, pp.
  954--960.

\bibitem{yu_polynomial_2017}
\BIBentryALTinterwordspacing
Q.~Yu, M.~A. Maddah-Ali, and A.~S. Avestimehr, ``Polynomial {Codes}: an
  {Optimal} {Design} for {High}-{Dimensional} {Coded} {Matrix}
  {Multiplication},'' \emph{arXiv:1705.10464 [cs, math]}, May 2017, arXiv:
  1705.10464. [Online]. Available: \url{http://arxiv.org/abs/1705.10464}
\BIBentrySTDinterwordspacing

\bibitem{li_coded_2015}
S.~Li, M.~A. Maddah-Ali, and A.~S. Avestimehr, ``Coded {MapReduce},'' in
  \emph{2015 53rd {Annual} {Allerton} {Conference} on {Communication},
  {Control}, and {Computing} ({Allerton})}, Sep. 2015, pp. 964--971.

\bibitem{tandon_gradient_2016}
\BIBentryALTinterwordspacing
R.~Tandon, Q.~Lei, A.~G. Dimakis, and N.~Karampatziakis, ``Gradient {Coding},''
  \emph{arXiv:1612.03301 [cs, math, stat]}, Dec. 2016, arXiv: 1612.03301.
  [Online]. Available: \url{http://arxiv.org/abs/1612.03301}
\BIBentrySTDinterwordspacing

\bibitem{dutta_short-dot}
S.~Dutta, V.~Cadambe, and P.~Grover, ``Short-{Dot}: {Computing} {Large}
  {Linear} {Transforms} {Distributedly} {Using} {Coded} {Short} {Dot}
  {Products},'' in \emph{Advances in {Neural} {Information} {Processing}
  {Systems} 29}, D.~D. Lee, M.~Sugiyama, U.~V. Luxburg, I.~Guyon, and
  R.~Garnett, Eds.\hskip 1em plus 0.5em minus 0.4em\relax Curran Associates,
  Inc., 2016, pp. 2100--2108.

\bibitem{Karakus_2017}
C.~Karakus, Y.~Sun, and S.~Diggavi, ``Encoded distributed optimization,'' in
  \emph{International Symposium on Information Theory (ISIT)}, 2017.

\bibitem{parikh_proximal_2014}
N.~Parikh and S.~Boyd, ``Proximal {Algorithms},'' \emph{Found. Trends Optim.},
  vol.~1, no.~3, pp. 127--239, Jan. 2014.

\bibitem{puri_multiple_1999}
R.~Puri and K.~Ramchandran, ``Multiple description source coding using forward
  error correction codes,'' in \emph{Conference {Record} of the
  {Thirty}-{Third} {Asilomar} {Conference} on {Signals}, {Systems}, and
  {Computers}}, vol.~1, Oct. 1999, pp. 342--346 vol.1.

\end{thebibliography}
%%
%% where we here have assume the existence of the files
%% definitions.bib and bibliofile.bib.
%% BibTeX documentation can be obtained at:
%% http://www.ctan.org/tex-archive/biblio/bibtex/contrib/doc/
%%
%%
%%
%% Or manual references (pay attention to consistency!):
%\begin{thebibliography}{1}
%\bibitem{shannon1948}
%  C.~E. Shannon, ``A mathematical theory of communication,''
%  \emph{Bell System Techn. J.}, vol.~27, pp. 379--423 and 623--656,
%  Jul. and Oct. 1948. 
%\end{thebibliography}

\end{document}